\documentclass[runningheads]{llncs}
\usepackage{graphicx}
\usepackage{hyperref}
\usepackage[T1]{fontenc}
\usepackage[utf8]{inputenc}
\usepackage[sort,compress]{cite}
\usepackage{amsmath}
\usepackage{amssymb}
\usepackage{xcolor}
\usepackage{enumitem}
\usepackage{comment}
\usepackage[strings]{underscore}

\begin{document}
\title{Four short stories on surprising algorithmic uses of treewidth\thanks{This research is a part of a project that has received funding from the European Research Council (ERC) under the European Union's Horizon 2020 research and innovation programme under grant agreement SYSTEMATICGRAPH (No.~725978).}}

\author{D\'aniel Marx}
\authorrunning{D.~Marx}
\institute{Max Planck Institute for Informatics\\ Saarland Informatics Campus, Germany\\
  \email{dmarx@mpi-inf.mpg.de} \\ \ \\
  {\em Dedicated to Hans L.~Bodlaender on the occasion of his 60th birthday.}
}

\maketitle

\global\long\def\Aut#1{\mathsf{Aut}(#1)}
\global\long\def\Emb#1#2{\mathsf{Emb}(#1\to#2)}
\global\long\def\StrEmb#1#2{\mathsf{StrEmb}(#1\to#2)}
\global\long\def\Hom#1#2{\mathsf{Hom}(#1\to#2)}
\global\long\def\Sub#1#2{\mathsf{Sub}(#1\to#2)}

  \newcommand{\rhofine}{\rho_0}
 \global\long\def\StrEmbProb#1{\mathsf{\#StrEmb}(#1)}
\global\long\def\EmbProb#1{\mathsf{\#Emb}(#1)}
\global\long\def\EmbNum{\mathsf{\#Emb}}
\global\long\def\IndProb#1{\mathsf{\#Ind}(#1)}
\global\long\def\IndPropProb#1{\mathsf{\#IndProp}(#1)}
\global\long\def\HomProb#1{\mathsf{\#Hom}(#1)}
\global\long\def\HomNum{\mathsf{\#Hom}}
\global\long\def\contract#1#2{#1 / #2}

\begin{abstract}
This article briefly describes four algorithmic problems where the notion of treewidth is very useful. Even though the problems themselves have nothing to do with treewidth, it turns out that combining known results on treewidth allows us to easily describe very clean and high-level algorithms.
\end{abstract}


\section{Introduction}

While the definition of treewidth may seem very technical at first sight, the naturality of treewidth is witnessed by the fact that it was introduced independently at least three times with equivalent definitions by different authors \cite{halin72,DBLP:journals/jct/BerteleB73,DBLP:journals/jct/RobertsonS84}. One may arrive to the study of treewidth from various directions and justify its importance with different arguments. One can, for example, argue that graphs of low treewidth (or some generalization of it) appear naturally in certain applications \cite{DBLP:conf/amw/FischlGLP19,DBLP:journals/pvldb/BonifatiMT17,DBLP:conf/icdt/ManiuSJ19,DBLP:journals/iandc/Thorup98}, hence algorithms for such graphs could be of practical interest. Or one could say that algorithms on bounded-treewidth graphs are based on the fundamental idea of recursively splitting the problem along small separators, and the study of treewidth is a good formalization of the study of this basic principle. But perhaps the nicest and most surprising reason for arriving at this notion is when the original goal has nothing to do with treewidth, but suddenly treewidth appears as the right theoretical tool for handling the problem. This article contains four such ``war stories,'' where the notion of treewidth and algorithms for bounded-treewidth graphs give very elegant solutions, which are sometimes in fact more efficient than those that were obtained earlier by involved and problem-specific techniques.

The four stories below are intentionally kept very brief in order to highlight the conceptual simplicity of the arguments. The aim is to show how certain high-level results can be combined in a clean way to achieve our goals. The detailed discussions or proofs of the results we are building on are beyond the scope of this article.
Later in this volume, the article of Marcin Pilipczuk contains more advanced examples of algorithmic use of treewidth bounds \cite{marcin-survey}.

\section{Bidimensionality}

Restricting an algorithmic problem to a certain family of graphs can make it easier than trying to solve it in general on every possible graph. A large part of the literature on algorithmic graph theory concerns algorithms for restricted classes of graphs that are of practical or theoretical significance. Restriction to planar graphs are studied both because of their interesting mathematical properties and as a starting point for modelling, e.g., road networks or 2D geometric problems.

From the viewpoint of polynomial-time solvability vs.~NP-hardness, the restriction to planarity does not seem to make the problem significantly easier. Most of the classic NP-hard problems (e.g., \textsc{3-Coloring}, \textsc{Maximum Indepenent Set}, \textsc{Hamiltonian Cycle}, etc.) remain NP-hard on planar graphs. The situation is very different from the viewpoint of parameterized complexity. Many of the basic problems that are W[1]-hard on general graphs turn out to be FPT on planar graphs. In fact, it took some time to arrive to the first relatively simple and natural problems that are W[1]-hard on planar graphs \cite{DBLP:journals/mst/CaiFJR07,DBLP:conf/iwpec/BodlaenderLP09}.

The restriction to planarity can help even for problems that are already FPT for general graphs. One of the main goals of the area of parameterized algorithms is to design algorithms with running time $f(k)n^{O(1)}$ such that the dependence $f(k)$ on the parameter is a function that grows as slowly as possible. For many of the fundamental problems studied in parameterized algorithms (e.g., \textsc{Vertex Cover}, \textsc{Feedback Vertex Set}, \textsc{$k$-Path}, \textsc{Odd Cycle Transversal}), algorithms with running time $2^{O(k)}n^{O(1)}$ are known. Furthermore, it is very likely that this form of running time is optimal: it is known that, under the Exponential Time Hypothesis (ETH) \cite{MR1894519,ImpagliazzoP01}, no algorithm with running time $2^{o(k)}n^{O(1)}$ exists for these problems. When restricted to planar graphs, significantly better algorithms are known for many of these problems, typically with running times of the form $2^{O(\sqrt{k})}n^{O(1)}$ or $2^{O(\sqrt{k}\log k)}n^{O(1)}$. Below we show how a very clean argument based on treewidth delivers such agorithms for certain basic problems; for others, more involved problem-specific ideas are needed \cite{KleinM14,pst-kernel,PilipczukPSL13,FominLMPPS16,DBLP:conf/fsttcs/LokshtanovSW12,DBLP:journals/corr/AboulkerBHMT15,DBLP:journals/iandc/DornFLRS13}. The main argument we present here was described first by Fomin and Thilikos \cite{DBLP:journals/siamcomp/FominT06} (for the \textsc{Dominating Set} problem) and was further developed under the name ``bidimensionality'' (see, e.g., \cite{DBLP:journals/talg/DemaineFHT05,DBLP:journals/jacm/DemaineFHT05,DBLP:journals/algorithmica/DemaineHT05,DBLP:journals/siamdm/DemaineFHT04}).

Let us consider the \textsc{$k$-Path} problem as our running example: given a (planar) graph $G$ and an integer $k$, we have to decide if $G$ contains a simple path on $k$ vertices. Let us first note that \textsc{$k$-Path} is FPT parameterized by the treewidth $w$ of the input graph $G$. More precisely, standard dynamic programming techniques give $2^{O(w\log w)}n^{O(1)}$ running time, while more sophisticated arguments are needed to obtain $2^{O(w)}n^{O(1)}$ time \cite{DBLP:journals/jcss/DornFT12,DBLP:journals/algorithmica/DornPBF10,DBLP:conf/stacs/Dorn10,DBLP:journals/dam/Dorn10,DBLP:journals/iandc/BodlaenderCKN15,DBLP:conf/focs/CyganNPPRW11,DBLP:journals/jacm/FominLPS16} (note that some of these algorithms are randomized  and some of these algorithms work only on planar graphs).
\begin{theorem}\label{th:kpathtw}
\textsc{$k$-Path} can be solved in time $2^{O(w)}n^{O(1)}$ if a tree decomposition of width $w$ is given in the input.
\end{theorem}
The second ingredient that we need is the Planar Excluded Grid Theorem \cite{MR1305057,DBLP:journals/algorithmica/GuT12}. A {\em minor} of a graph $G$ is a graph $H$ that is obtained by a sequence of vertex deletions, edge deletions, and edge contractions. A {\em $k\times k$ grid} is a graph with vertex set $[k]\times [k]$, where vertices $(x,y)$ and $(x',y')$ are adjacent if and only if $|x-x'|+|y-y'|=1$. The following theorem states that, in a very tight sense, the existence of a grid minor is the canonical reason why a planar graph has large treewidth:
\begin{theorem}[Planar Excluded Grid Theorem]\label{th:gridminor}
Every planar graph with treewidth at least $4.5k$ has a $k\times k$ grid minor.
\end{theorem}
In particular, Theorem~\ref{th:gridminor} implies that an $n$-vertex planar graph has treewidth $O(\sqrt{n})$: it certainly cannot contain a grid minor larger than $\sqrt{n}\times \sqrt{n}$.

Finally, we have to make two simple observations about the \textsc{$k$-Path} problem:
\begin{enumerate}
\item[(1)] The $k\times k$ grid contains a path on $k^2$ vertices: imagine a ``snake'' that visits the rows one after the other.
\item[(2)] If $H$ is a minor of $G$, then the length of the longest path in $H$ is not larger than in $G$. This can be proved by verifying that none of vertex deletion, edge deletion, or edge contraction can increase the length of the longest path.
\end{enumerate}
Now the claimed algorithm can be obtained by putting together these ingredients using a win/win approach. For simplicity, we describe an algorithm for the decision version of the problem where only a YES/NO answer has to be returned.
\begin{theorem}
\textsc{$k$-Path} on planar graphs can be solved in time $2^{O(\sqrt{k})}n^{O(1)}$.
\end{theorem}
\begin{proof}
  Let $w:=4.5\lceil\sqrt k \rceil$. If $G$ is a graph with treewidth at least $w$, then Theorem~\ref{th:gridminor} implies that $G$ contains a $\lceil \sqrt k \rceil \times \lceil \sqrt k \rceil $ grid minor $H$. Then the first observation above shows that $H$ contains a path on $k$ vertices and the second observation shows that $G$ also contains a path on $k$ vertices. Therefore, we can conclude that if the input graph $G$ has treewidth at least $w$, then it is a YES-instance: it surely contains a path on $k$ vertices.

  The algorithm proceeds as follows. First, we compute an (approximate) tree decomposition of $G$. For this purpose, it is convenient to use the algorithm  of Bodlaender et al.~\cite{DBLP:journals/siamcomp/BodlaenderDDFLP16}, which, given an integer $w$ and a graph $G$, in time $2^{O(w)}\cdot n=2^{O(\sqrt{k})}\cdot n$ either correctly states that treewidth of $G$ is larger than $w$, or gives a tree decomposition of width at most $5w+4$. We can complete the computation in both cases:
  \begin{itemize}
  \item If the algorithm states that $G$ has treewidth larger than $w$, then, as we have seen above, the answer is YES.
    \item If the algorithm returns a tree decomposition of width at most $5w+4=O(\sqrt{k})$, then we can invoke Theorem~\ref{th:kpathtw} to decide the existence of a path on $k$ vertices and return YES or NO accordingly. The running time is $2^{O(w)}n^{O(1)}=2^{O(\sqrt{k})}n^{O(1)}$, as required.
    \end{itemize}
Thus we have an algorithm that returns a correct YES/NO-answer in time $2^{O(\sqrt{k})}\cdot n^{O(1)}$.
\end{proof}
The same argument works for \textsc{Feedback Vertex Set} and \textsc{Vertex Cover}. Only the analogs of the two observations (1) and (2) need to be verified: the optimum value is $\Omega(k^2)$ on the $k\times k$ grid and that the minor operation cannot increase the optimum value.  A variant of the argument, based on contractions instead of minors, can give algorithms for \textsc{Independent Set} and \textsc{Dominating Set}. There are also less straighforward uses of Theorem~\ref{th:gridminor}, where it is invoked not on the input graph itself, but on some auxilliary graph defined in a nonobvious way; see the article of Marcin Pilipczuk later in this volume for some examples \cite{marcin-survey}.

\section{Exponential-time algorithms for graphs of maximum degree 3}

If the task is to find a subset of vertices satisfying certain properties, then  we can typically solve the problem in time $2^n\cdot n^{O(1)}$ on graphs with $n$ vertices by enumerating every subset. For many problems, it is easy to improve on this brute force algorithm. For example, in the case of the \textsc{Maximum Independent Set} problem (for graphs with arbitrarily large degree), there is a simple textbook example of an improved branching algorithm that beats the $2^n\cdot n^{O(1)}$ running time. As long as there is a vertex $v$ of degree at least 3, branch into two directions: either the solution avoids $v$ (in which case we can remove $v$, decreasing the size of the graph by 1) or it contains $v$ (in which case we can remove $v$ and its neighbors from the problem, decreasing the size of the graph by at least 4 vertices). The problem can be solved in polynomial time if every vertex has degree at most $2$. Analyzing the algorithm shows that its running time is $1.3803^n\cdot n^{O(1)}$. Further improvements are possible with more and more involved techniques \cite{DBLP:journals/siamcomp/TarjanT77,DBLP:journals/algorithmica/BourgeoisEPR12,DBLP:conf/fsttcs/KneisLR09,Jian1986847,DBLP:journals/jal/Robson86,DBLP:journals/jacm/FominGK09} with the current best algorithm having running time $1.1996^n\cdot n^{O(1)}$ \cite{DBLP:journals/iandc/XiaoN17}. Similar ``races'' for the best exponential-time algorithm are known for many other problems \cite{DBLP:series/txtcs/FominK10}. Let us remark that for some problems just beating the trivial $2^n\cdot n^{O(1)}$ running time is already highly nontrivial \cite{DBLP:journals/algorithmica/BliznetsFPV16,DBLP:conf/ictcs/Razgon07,DBLP:journals/algorithmica/CyganPPW14a}.

For the \textsc{Maximum Independent Set} problem on graphs of maximum degree 3, the current best algorithm has running time $1.0836^n\cdot n^{O(1)}$ \cite{DBLP:journals/tcs/XiaoN13}. Here we would like to highlight an earlier, less efficient algorithm that can be explained using the notion of treewidth very easily. Fomin and H\o ie \cite{DBLP:journals/ipl/FominH06} proved, using an earlier result of Monien and Preis \cite{DBLP:conf/mfcs/MonienP01}, that the pathwidth (and hence the treewidth) of an $n$-vertex graph with maximum degree 3 is essentially at most $n/6$. More precisely:

\begin{theorem}[Fomin and H\o ie~\cite{DBLP:journals/ipl/FominH06}]\label{th:cubicpathwidth}
For any $\epsilon>0$, there is an integer $n_\epsilon$ such that the pathwidth of any graph on $n>n_\epsilon$ vertices and maximum degree at most 3 is at most $(1/6+\epsilon)n$.
\end{theorem}
Together with the fact that a \textsc{Maximum Independent Set} on an $n$-vertex  graph can be solved in time $2^w\cdot n^{O(1)}$ if a tree decomposition of width $w$ is given, it follows that the problem can be solved in time $2^{n/6}\cdot n^{O(1)}=1.1225^n\cdot n^{O(1)}$. The running time obtained as a simple consequence of this pathwidth bound was better than some earlier work at that time \cite{DBLP:conf/soda/Beigel99,DBLP:journals/algorithmica/ChenKX05}, but since then improved algorithms with more complicated and problem-specific arguments were found for this problem \cite{DBLP:journals/tcs/XiaoN13,DBLP:conf/iwpec/BourgeoisEP08,DBLP:journals/algorithmica/BourgeoisEPR12,DBLP:journals/jda/Razgon09}. In a similar way, algorithms for \textsc{Minimum Dominating Set} and \textsc{Max Cut} follow immediately from Theorem~\ref{th:cubicpathwidth}, which were better than some of the algorithms found by earlier problem specific techniques \cite{DBLP:journals/ipl/FominH06}.

\section{Finding and counting permutation patterns}

Interesting combinatorial and algorithmic problems can be defined on permutations and on the patterns they contain or avoid. A {\em permutation} of length $n$ is a bijection $\pi:[n]\to [n]$; typically we describe permutations by the sequence $(\pi(1), \pi(2), \ldots, \pi(n))$. We say that a permutation $\sigma$ of length $n$ {\em contains} a permutation $\pi$ of length $k$ if there is a mapping $f:[k]\to [n]$ such that $f(1)<f(2)<\dots< f(k)$ and $\pi(i)<\pi(j)$ if and only if $\sigma(f(i))< \sigma(f(j))$. That is, $\sigma$ contains $\pi$ if the sequence $(\pi(1), \ldots, \pi(k))$ can be mapped to a subsequence of $(\sigma(1),\dots, \sigma(n))$ in a way that preserves the relative order of the values. As an example, the permutation  $( 3, 4, 5, 2, 1, 7, 8, 6)$ contains the permutation $(2, 1, 3, 4)$ (e.g., by the mapping $(f(1),f(2),f(3),f(4))=(1,4,6,7)$), but it does not contain the permutation $(4, 3, 2, 1)$. Observe that the permutations {\em not} containing $(1, 2)$ are exactly the decreasing sequences, while the permutations {\em not} containing $(2, 1)$ are exactly the increasing sequences. As shown by Knuth~\cite[§\,2.2.1]{knuth68}, the permutations avoiding $(2, 3, 1)$ are exactly the permutations sortable by a single stack. From the extremal combinatorics point of view, a very natural question is to bound the number of permutations of length $n$ avoiding a fixed permutation $\pi$. Marcus and Tardos \cite{MT} proved a long-standing conjecture of Stanley and Wilf\footnote{Marcus and Tardos \cite{MT} mentions that the conjecture was formulated around 1992 (but it is hard to find a citable source) and the PhD thesis of Julian West is an even earlier source \cite{west-phd}.} by showing that for every fixed permutation $\pi$, there is a constant $c(\pi)$ such that the number of permutations of length $n$ avoiding $\pi$ is at most $2^{c(\pi)\cdot n}$. This has to be contrasted with the fact that the total number of permutations of length $n$ is $n!=2^{O(n\log n)}$.

From the algorithmic point of view, perhaps the most fundamental question is testing for containment: given a permutation $\sigma$ of length $n$ and a permutation $\pi$ of length $k$, does $\sigma$ contain $\pi$? The problem is often called \textsc{Permutation Pattern Matching} and is known to be NP-hard \cite{BBL}, but of course can be solved in time $O(n^{k})$ by brute force. Albert et al.~\cite{Albert_algo} improved this to $O(n^{2/3k+1})$ time,  Ahal and Rabinovich \cite{Ahal} further improved it to $n^{0.47k+o(k)}$ time, and Berendsohn et al.~\cite{Berendsohn-ipec2019} gave an $n^{0.25k+o(k)}$ time algorithm. Guillemot and Marx~\cite{GM} showed that \textsc{Permutation Pattern Matching} can be solved in time $2^{O(k^2\log k)}\cdot n$, that is, it is fixed-parameter tractable (FPT) parameterized by the length of $\pi$.

Even though the problem is FPT, algorithms with running time $n^{ck}$ can be still interesting for two reasons. First, if $k$ is fairly large, say, $\Omega(\log n)$, then  $2^{O(k^2\log k)}\cdot n$ is actually worse than $n^{O(k)}$. Thus unless we have $2^{O(k)}\cdot n^{O(1)}$ FPT algorithms for the problem, we need different type of algorithms to understand the complexity of the problem in the regime where $k$ is large. Second, the $n^{ck}$ time algorithms \cite{Albert_algo, Ahal,Berendsohn-ipec2019} can be easily modified to count the total number of solutions, while the FPT algorithm of Guillemot and Marx~\cite{GM} returns only a single solution. This is not just a shortcoming of the presentation \cite{GM}: the FPT algorithm contains a step where a certain structure is discovered that guarantees that every permutation of length $k$ appears in $\sigma$. Then the algorithm stops and does not look for any further occurences of $\pi$. Furthermore, it is unlikely that the algorithm can be extended to a counting version: Berendsohn et al.~\cite{Berendsohn-ipec2019} proved that the counting problem is \#W[1]-hard.

The $n^{ck}$ algorithms for \textsc{Permutation Pattern Matching} \cite{Albert_algo, Ahal,Berendsohn-ipec2019} are implicitly or explicitly based on dynamic programming on a certain tree decomposition. Here we follow the presentation of Berendsohn et al.~\cite{Berendsohn-ipec2019}, where it is shown how high-level arguments and previous results on treewidth can be combined to obtain an $n^{k/3+o(k)}$ time in a very clean way (a further improvement, based on a technical idea of Cygan et al.~\cite{Cygan}, reduces the running time to $n^{0.25k+o(k)}$ \cite{Berendsohn-ipec2019}).

A permutation $\pi:[k]\to [k]$ can be seen as a $k$-element point set $S_\pi=\{(i,\pi(i)) \mid i\in [k]\}$ (see Figure~\ref{fig1}). With this interpretation, $\sigma$ contains $\pi$ if $S_\pi$ can be mapped to a subset of $S_\sigma$ in a way that the mapping preserves the relative ordering of any two points along both the horizontal axis and the vertical axis. For a point $p\in S_\pi$, we will denote by $p.x$ and $p.y$ the first and second coordinates of $p$, respectively. For each point $(x,y)\in S_\pi$, we define the four neighbors of $(x,y)$ as follows:
\begin{eqnarray*}
N^R((x,y)) & = & (x+1,~ \pi(x+1)), \\
N^L((x,y)) & = & (x-1,~ \pi(x-1)), \\
N^U((x,y)) & = & (\pi^{-1}(y+1),~ y+1), \\
N^D((x,y)) & = & (\pi^{-1}(y-1),~ y-1). 
\end{eqnarray*}

\begin{figure}
\begin{center}
\includegraphics[width=0.4\linewidth]{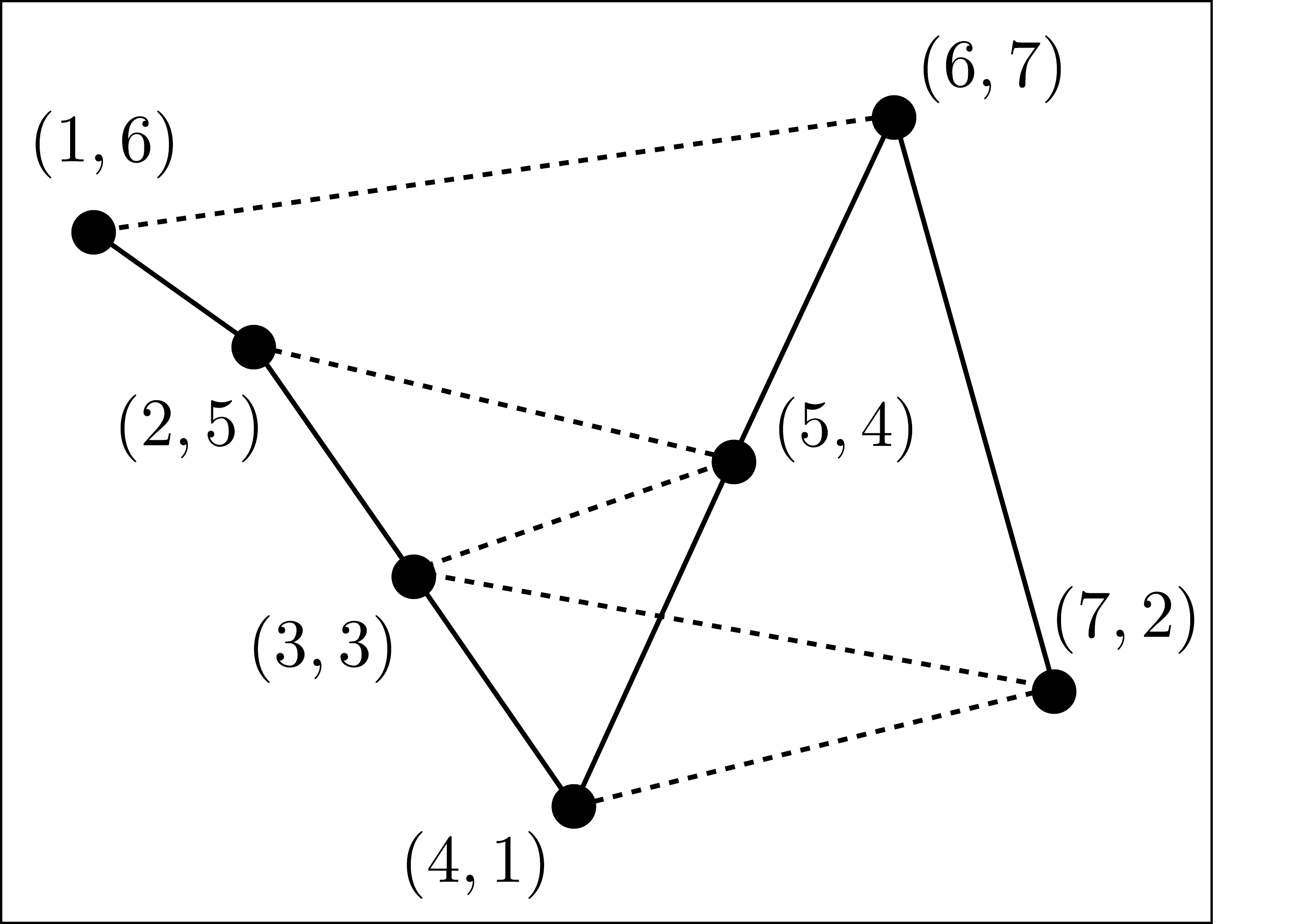}
 \caption{Permutation $\pi = (6,5,3,1,4,7,2)$ and its incidence graph $G_\pi$. Solid lines indicate neighbors by index (L-R), dashed lines indicate neighbors by value (U-D). Indices plotted on $x$-coordinate, values plotted on $y$-coordinate.\label{fig1}}
\end{center}
\end{figure}

The superscripts $R$, $L$, $U$, $D$ are meant to evoke the directions \emph{right}, \emph{left}, \emph{up}, \emph{down}, when plotting $S_\sigma$ in the plane. That is,  if we start sweeping the vertical line going through $(x,y)$ to the \underline{\em R}\hspace{1pt}ight, then $N^R((x,y))$ is the next point that we meet, and similarly with the other directions. Note that some neighbors of a point may coincide.

The {\em incidence graph} $G_\pi$  of $\pi$ is a graph on $S_\pi$ where each point is connected to its four neighbors (when defined). It is easy to see that $G_\pi$ is the union of two Hamiltonian paths on the same set $S_\pi$ of vertices, with one path going in the left-right direction in the plane, while the other path going in the top-bottom direction.

The key lemma that allows a clean abstraction of the problem is the following characterization of solutions.

\begin{lemma}\label{lem:perm}
  Let $\sigma:[n]\to [n]$ and $\pi:[k]\to [k]$ be two permutations. Then $\sigma$ contains $\pi$ if and only if there is a function $f:S_\pi\to S_\sigma$ such that for every $p\in S_\pi$
\begin{eqnarray}
f(N^L(p)).x & < ~~f(p).x  & < ~~f(N^R(p)).x, \mbox{~~and}\\
f(N^D(p)).y & < ~~f(p).y & < ~~f(N^U(p)).y, 
\end{eqnarray}
whenever the corresponding neighbor of $p$ is defined. 
\end{lemma}

It is not very difficult to prove Lemma~\ref{lem:perm} using the definitions and
we can also see that the functions $f$ satisfying the requirements of Lemma~\ref{lem:perm} are in one to one correspondence with the occurrences of $\pi$ in $\sigma$.
The inequalities in the first line ensure that the mapping of points represent the left-to-right ordering, while the inequalities in the second line handle the top-to-bottom ordering.
The key observation is that even though we require these inequalities only between neighbors in $G_{\pi}$, it follows as consequence that every pairwise inequality in the definition of containment holds. For example, if $\pi(i)<\pi(j)$, then $(j,\pi(j))$ can be reached from $(i,\pi(i))$ by going through a sequence of U-neighbors, hence a sequence of inequalities ensure that the second coordinate of $f((i,\pi(i))$ is less than the second coordiante of $f((j,\pi(j)))$.

Readers familar with the notion of Constraint Satisfaction Problems (CSPs) may recognize that Lemma~\ref{lem:perm} cleanly transforms the problem into a binary constraint satisfaction problem.  A \emph{binary CSP} instance is a triplet $(V,D,C)$, where $V$ is a set of variables, $D$ is a set of admissible values (the \emph{domain}), and $C$ is a set of constraints $C = \{c_1, \dots, c_m\}$, where each constraint $c_i$ is of the form $((x,y),R)$, where $x,y \in V$, and $R \subseteq D^2$ is a binary relation.
A solution of the CSP instance is a function $f: V \rightarrow D$ (i.e., an assignment of admissible values to the variables), such that for each constraint $c_i = ((x_i,y_i),R_i)$, the pair of assigned values $(f(x_i),f(y_i))$ is contained in $R_i$.

The \emph{constraint graph} of the binary CSP instance (also known as \emph{primal graph} or \emph{Gaifman graph}) is a graph whose vertices are the variables $V$ and whose edges connect all pairs of variables that occur together in a constraint. Low treewidth of the constraint graph can be exploited for an efficient solution of the problem:

\begin{theorem}[\!\!\cite{treewidthCSP1, treewidthCSP2}] \label{thmtr1}
A binary CSP instance $(V,D,C)$ can be solved in time $O(|D|^{t+1})$ where $t$ is the {treewidth} of the constraint graph.
\end{theorem}

To view the \textsc{Permutation Pattern Matching} problem as a binary CSP instance, let $V=S_\pi$ be the set of variables and let $D=S_\sigma$ be the domain. Then we want to find a function $f$ that satisfies the inequalities in Lemma~\ref{lem:perm}. Each inequality is a binary constraint between $p$ and $N^\alpha(p)$ for some $\alpha\in \{L,R,D,U\}$, restricting the possible combination of values that $f(p)$ and $f(N^\alpha(p))$ can take. Thus we end up with a CSP instance on $k$ variables, domain size $n$, and whose constraint graph is exactly $G_\pi$.

In order to invoke Theorem~\ref{thmtr1} on this instance, we need to bound the treewidth of $G_\pi$. Recall that $G_\pi$ has $k$ vertices and maximum degree $4$. By splitting each degree-4 vertex into two degree-3 vertices connected by an edge, we can create a graph $G'_\pi$ that has at most $2k$ vertices, maximum degree 3, and $G_\pi$ is a minor of $G'_\pi$. Then Theorem~\ref{th:cubicpathwidth} shows that $G'_\pi$ has treewidth $2k/6+o(k)=k/3+o(k)$ and $G_\pi$ being a minor of $G'_\pi$ shows that the same bound holds for $G_\pi$ as well. Therefore, we can conclude that Theorem~\ref{thmtr1} solves the instance in time $n^{k/3+o(k)}$. It is not difficult to modify the algorithm to count the number of solutions. Therefore, the combination of an easy observation (Lemma~\ref{lem:perm}), a combinatorial treewidth bound (Theorem~\ref{th:cubicpathwidth}), and a known general algorithm (Theorem~\ref{thmtr1}) solves the problem in a very clean way.

In Lemma~\ref{lem:perm}, the functions $f$ satisfying the requirements are in one to one correspondence with the occurences of $\pi$ in $\sigma$ and Theorem~\ref{thmtr1} can be extended to a counting version. Theorem~\ref{th:cubicpathwidth} is purely combinatorial, thus it is of course irrelevant if we are using it for the decision or the counting problem. Thus the same algorithmic idea goes through.
\begin{theorem}[Berendsohn et al.~\cite{Berendsohn-ipec2019}]
Given a length-$k$ permutation $\pi$ and length-$n$ permutation $\sigma$, the number of occurrences of $\pi$ in $\sigma$ can be counted in time $n^{k/3+o(k)}$.
  \end{theorem}

\section{Counting subgraphs}
It is a well-known phenomenon in theoretical computer science that in many cases finding a solution is easier than counting the number of all solutions. For example, it can be checked in polynomial time if a bipartite graph contains a perfect matching, but the seminal result of Valiant shows that counting the number of perfect matchings is \#P-hard and hence unlikely to be polynomial-time solvable \cite{DBLP:journals/tcs/Valiant79}. By now, many other examples of hard counting problems are known.

Flum and Grohe \cite{DBLP:journals/siamcomp/FlumG04} started the investigation of the complexity of counting in the setting of parameterized complexity. They introduced the notion of \#W[1]-hardness to give evidence that certain parameterized counting problems are unlikely to be FPT. As a highly nontrivial example, they considered the \textsc{$k$-Path} problem: the decision version is known to be FPT by various techniques \cite{DBLP:journals/jacm/AlonYZ95,DBLP:journals/jacm/FominLPS16}, but they showed that the counting version of the problem is \#W[1]-hard. In the same paper, they asked as an open question whether the counting version of the polynomial-time solvable \textsc{$k$-Matching} problem is FPT. This question was resolved in the negative by the \#W[1]-hardness proof of Curticapean~\cite{DBLP:conf/icalp/Curticapean13}, which used heavy algebraic machinery, and by the later simpler proof given by Curticapean and Marx \cite{DBLP:conf/focs/CurticapeanM14}. More recently, Dell et al. \cite{DBLP:conf/stoc/CurticapeanDM17} described and exploited a connection beween subgraph counting and homomorphism counting problems. This connection can be useful in two different ways: it gives new subgraph-counting algorithms by reducing it to homomorphism-counting problems, and gives hardness results for subgraph counting (including new and clean \#W[1]-hardness proofs of \textsc{$k$-Matching} and \textsc{$k$-Path}) based on our understanding of the complexity of counting homomorphisms. Below we give an example of the algorithmic use of this connection.

Given the \#W[1]-hardness of \textsc{$k$-Path}, we cannot hope for an FPT algorithm solving the problem. But it is still an interesting question whether we can improve on the trivial $n^{k+O(1)}$ time brute force algorithm.  The ``meet in the middle'' approach can be used to improve this to $n^{k/2+O(1)}$ time \cite{koutis2009limits,bjorklund2009counting}, which was further improved by Björklund et al.~\cite{fasterthanmeetinthemiddle} to $n^{0.455k+O(1)}$. Here we describe an algorithm with running time $k^{O(k)}\cdot n^{0.174k+o(k)}$, which has a much smaller exponent for a fixed $k$ and at the same time conceptually much simpler.

Let us first review some basic background on homomorphisms. A {\em homomorphism} from graph $H$ to graph $G$ is a mapping $f:V(H)\to V(G)$ such that for every edge $uv\in E(H)$, we have
$f(u)f(v)\in E(G)$. We will denote by $\#\Hom HG$ the number of homomorphisms from $H$ to $G$. Given a tree decomposition of $H$, standard dynamic programming techniques can be used to compute the number of homomorphisms from $H$ to a given graph $G$.

\begin{theorem}[Díaz et al.~\cite{DBLP:journals/tcs/DiazST02}]\label{th:homcount}
Given graphs $H$ and $G$,  $\#\Hom HG$ can be computed in time $(|V(H)|+|V(G)|)^{w+O(1)}$, where $w$ is the treewidth of $H$.
\end{theorem}
Note that the algorithm of Theorem~\ref{th:homcount} does not need a decomposition of $H$, as it can be found in time $|V(H)|^{c+O(1)}$.

  A homomorphism $f:V(H)\to V(G)$ is {\em injective} if $f(u)\neq f(v)$ for any two distinct $u,v\in V(H)$; let $\#\Emb HG$ denote the number of such homomorphisms.
Let us denote by $\#\Sub HG$ the number of subgraphs of $G$ that are isomorphic to $H$.  It is well known and easy to see that $\#\Emb HG=\# \Sub HG \cdot \# \Aut H$, where $\# \Aut H = \#\Emb HH$ is the number of automorphisms of the graph $H$. Therefore, for a fixed $H$, computing $\# \Sub HG$ is essentially equivalent to computing $\# \Emb HG$, the number of injective homomorphisms. In order to explain the connection between counting homomorphisms and subgraphs, it will be more convenient to work with $\# \Emb HG$ than with $\# \Sub HG$, as the former is already defined in terms of homomorphisms.

Of course, not every homomorpism from $H$ to $G$ is injective, the images of some vertices may coincide.
For example, if $H$ is the 4-cycle on vertices ${1,2,3,4}$, then a homomorphism from $H$ to a loopless graph $G$ either (1) is injective, (2) identifies $1$ with $3$, (3) identifies $2$ with $4$, (4) identifies $1$ with $3$, and $2$ with $4$. In case (1), the image of $H$ is a 4-cycle; in cases (2) and (3), the image of $H$ is the path $P_3$ on three vertices; and in case (4), the image of $H$ is the path $P_2$ on two vertices. This shows that the following formula holds for the number of homomorphisms:
\[
\#\Hom{C_4}{G}=\#\Emb{C_4}{G}+2\cdot\#\Emb{P_3}{G}+\#\Emb{{P_2}}{G}.
\]

More generally, we can classify the homomorphisms according to which sets of vertices they identify.
To each homomorphism $h:V(G)\to V(H)$, we 
can associate a partition $\rho_h$ of $V(H)$ with the meaning that, for every $u,v\in V(H)$, we have $h(u)=h(v)$ if and only $u$ and $v$ are in the same block of $\rho$.  For a partition $\rho$ of $V(H)$, let $\contract H\rho$ be the {\em quotient graph} obtained by consolidating each block of $\rho$ into a single vertex. The key observation is that the homomorphisms from $H$ to $G$ having type $\rho$ are in one-to-one correspondence with the injective homomorphisms from $\contract H\rho$ to $G$. Therefore, we can express the number of homomorphisms from $H$ to $G$ as
\begin{equation}
  \#\Hom{H}{G}=\sum_{\rho} \#\Emb{\contract H\rho}{G},\label{eq:hompartemb}
\end{equation}
  where the sum ranges over every partition $\rho$ of $V(H)$.

  Why is this useful for us? Observe that $H=\contract H\rho$ holds only for the partition $\rho_0$ where every block has size exactly one and $\contract H\rho$ has strictly fewer vertices for every other $\rho$.  Therefore, Eq.~\eqref{eq:hompartemb} can be written as
  \[
   \#\Hom{H}{G}=\#\Emb{H}{G}+\sum_{\rho\neq \rhofine} \#\Emb{\contract H\rho}{G},
 \]
 and hence
 \begin{equation}
   \#\Emb{H}{G}=\#\Hom{H}{G}-\sum_{\rho\neq \rhofine} \#\Emb{\contract H\rho}{G}. \label{eq:embparthom1}
 \end{equation}
 That is, Eq.~\eqref{eq:embparthom1} reduces the problem of computing $\#\Emb{H}{G}$ to the problem of computing $\#\Hom{H}{G}$ and to computing some number of $\#\Emb{\contract H\rho}{G}$ values, where $\contract H\rho$ has strictly fewer vertices than $|V(H)|$. Therefore, we can repeat the same argument and recursively replace each term  $\#\Emb{\contract H\rho}{G}$ with a $\HomNum$ term and some number of $\EmbNum$ terms. As the replacement strictly decreases the number of vertices in the $\EmbNum$ terms, eventually all these terms disappear, and we can express $\#\Emb{H}{G}$ as the linear combination of $\#\Hom{H'}{G}$ values for various graphs $H'$. This means that we can reduce the problem of computing $\# \Emb{H}{G}$ to computing certain homomorphism values.

 Which graphs $H'$ can appear in the $\#\Hom{H'}{G}$ terms when we express $\#\Emb{H}{G}$ this way? It is easy to see that the quotient graph of a quotient of $H$ is also a quotient graph of $H$. This means that every graph $H'$ appearing in this linear combination is a quotient graph of $H$.
Thus we can express $\#\Emb HG$ as
 \begin{align}
\label{eq: emb2hom-intro0}
  \#\Emb HG
  &=
  \sum_{\rho}
    \beta_{\rho,H} \cdot  \#\Hom {\contract H\rho}{G}
  \,,
 \end{align}
 where $\beta_{\rho,H}$ is a constant depending only on $\rho$ and $H$. The argument described above gives an algorithm for writing $\#\Emb HG$ in this form and for computing the constants $\beta_{\rho,H}$ (and the work of Lovász et al.~\cite{lovasz1967operations,borgs2006counting} gives more explicit formulas for these constants). Given this expression, we can reduce the problem of computing $\#\Emb HG$ to computing the values $\#\Hom {\contract H\rho}{G}$. If $H$ has $k$ vertices, then the sum ranges over $k^{O(k)}$ different partitions $\rho$. Therefore, if every $\contract H\rho$ has treewidth bounded by $c$, then invoking Theorem~\ref{th:homcount} for the computation of each $\#\Hom {\contract H\rho}{G}$ results in an algorithm with running time $k^{O(k)}\cdot n^{c+O(1)}$ for the computation of $\#\Emb HG$ (and hence of $\#\Sub HG$).

 These considerations show that bounding the running time of our algorithm essentially boils down to a bound on the maximum treewidth of $\contract H\rho$.
 The treewidth of $\contract H\rho$ can be much larger than the treewidth of $H$. For example, it is not difficult to see that if $H$ is a matching with $k$ independent edges, then we can obtain any connected graph with $k$ edges as $\contract H\rho$ for an appropriate partition $\rho$. However, this operation cannot increase the number of edges: if $H$ has $k$ edges, then $\contract H\rho$ has at most $k$ edges. We can use the following bound on the treewidth of graphs with at most $k$ edges:
 \begin{theorem}[\hspace{0.01em}\cite{Scott2007260,DBLP:journals/algorithmica/FominGSS09}]\label{th:kedgetw}
Every graph with at most $k$ edges has treewidth $0.174k+o(k)$.
\end{theorem}
This immediately gives an upper bound on the running time needed if $H$ has at most $k$ edges.
\begin{theorem}[Dell et al.~\cite{DBLP:conf/stoc/CurticapeanDM17}]
If $H$ has at most $k$ edges, then $\#\Emb HG$ and $\#\Sub HG$ can be computed in time $k^{O(k)}\cdot n^{0.174k+o(k)}$.
\end{theorem}
In particular, we obtain algorithms with running time $k^{O(k)}\cdot n^{0.174k+o(k)}$ if $H$ is a path with $k$ edges (the \textsc{$k$-Path} problem) or a matching with $k$ edges (the \textsc{$k$-Matching} problem). We want to emphasize that for a fixed $H$, the algorithm is very simple: it consists of invoking Theorem~\ref{th:homcount} for various graphs $H'=\contract H\rho$ and then taking a linear combination of these values. All the real work is done by the computation of the fixed constants $\beta_{\rho,H}$ and by the algorithm of Theorem~\ref{th:homcount} exploiting low treewidth and tree decompositions.

\bibliographystyle{splncs04}
\bibliography{marx}

\end{document}